\documentclass[sigconf]{acmart}
\AtBeginDocument{%
  \providecommand\BibTeX{{%
    \normalfont B\kern-0.5em{\scshape i\kern-0.25em b}\kern-0.8em\TeX}}}
\usepackage{cases}
\usepackage{stfloats}

\acmSubmissionID{519}

\newtheorem{theorem}{Theorem}
\usepackage{amsthm}
\usepackage[linesnumbered,ruled,vlined]{algorithm2e}

\SetKwInput{Parameters}{Parameters}                
\SetKwInput{KwInput}{Input}              

\begin{document}

\title{STSIR: Spatial Temporal Pandemic Model with Mobility Data }
\subtitle{A COVID-19 Case Study}
\author{Wang Pan}
\affiliation{%
  \institution{Tsinghua-Berkeley Shenzhen Institute, Tsinghua University}
  \institution{Pengcheng Laborotory}
  \country{China}}  
\email{pw18@mails.tsinghua.edu.cn}

\author{Qipu Deng}
\affiliation{%
  \institution{Tsinghua-Berkeley Shenzhen Institute, Tsinghua University}
  \institution{Pengcheng Laborotory}
  \country{China}}  
\email{dqp18@mails.tsinghua.edu.cn}

\author{Jiadong Li}
\affiliation{%
  \institution{Tsinghua-Berkeley Shenzhen Institute, Tsinghua University}
  \institution{Pengcheng Laborotory}
  \country{China}}  
\email{lijd18@mails.tsinghua.edu.cn}

\author{Zhi Wang}
\affiliation{%
  \institution{Department of Computer Science and Technology, Tsinghua University}
  \institution{Pengcheng Laborotory}
  \country{China}}
\email{wangzhi@sz.tsinghua.edu.cn}

\author{Wenwu Zhu}
\affiliation{%
  \institution{Department of Computer Science and Technology, Tsinghua University}
  \institution{Pengcheng Laborotory}
  \country{China}}
\email{wwzhu@tsinghua.edu.cn}


\begin{abstract}
With the outbreak of COVID-19, how to mitigate and suppress its spread is a big issue to the government. Department of public health need powerful models to model and predict the trend and scale of such pandemic. And models that could evaluate the effect of the public policy are also essential to the fight with the COVID-19.\\
A main limitation of existing models is that they can only evaluate the policy by calculating $R_0$ after infection happens instead of giving observable index. To tackle this, based on the transmission character of the COVID-19, we preposed a novel framework Spatial-Temporal-Susceptible-Infected-Removed (STSIR) model. In particular, we merged both intra-city and inter-city mobility index with the traditional SIR dynamics and make it a dynamic system. And we proved that the STSIR system is a closed system which makes the system self-consistent. And finally we proposed a Multi-Stage Simulated Annealing (MSSA) algorithm to find optimal parameter of the system. In our experiments, based on Baidu Mobility dataset \cite{baidu20}, and China pandemic dataset provided by Dingxiangyuan\cite{dxy}, our model can effectively predict the total scale of the pandemic and also gives clear policy analysis with observable index.
\end{abstract}

%

\bibliographystyle{unsrt}

\maketitle
\section{Introduction}
Building accurate models to predict and quantify the effect of infectious diseases is very important to today's public health. The classic SI, SIR and SEIR models \cite{kermack1927contribution} have been proposed and used for a long time. However, with the development of modern transportation, the range and the effect of infectious disease are widely amplified, which traditional SI, SIR and SEIR models become insufficient to describe the transmission characteristics of new infectious diseases accurately. For example, the outbreak of the new COVID-19 pandemic has caused nearly 20 million infections and hundreds of thousands of deaths nowadays. We have not expected so much loss that such disease can cause, which urge us to build better infectious disease models to analyze and monitor the pandemic's development. And we can expect to rely on such tools to help us understand the epidemic transmission characters and make correct and reasonable epidemic prevention policies to save lives.\\
In order to better fit and analyze the spread of such infectious diseases, various scholars have conducted different new methods. On one hand, several scholars add new transiting phases (eg. Asymptomatic, Hospitalized, etc.) in infectious models\cite{Gatto10484,Zhou2020medxivphase,Li2020.02.14.20022913,Wang2020Cell}. However, the method of simply adding phases has several shortcomings. Firstly, the abundant parameters make parameter calculation very laborious. Secondly, after adding multiple phases, the model focus on only simulates a single city or region, and cannot model the flow of people between multiple regions. The transmission characters between regions have been ignored. Thirdly, with the different definition of different phases, properly categorizing patients to those phases is difficult.\\
To make the model be able to model multiple regions, other scholars have added transportation statistics such as civil aviation and railways data to model pandemic evolution\cite{Zhou2020medxivphase,Ng2020NEJMTraveler,lau2020association,wu2020nowcasting,Liu2020.03.17.20037770}. Nevertheless, the information granularity of civil aviation and railway statistics are too coarse to give the intercity information accurately, such modeling can only give a rough pandemic trend, with low predicting accuracy. Such models considered the inter-region transmission from the origin of the pandemic outbreak to other major cities. But with the development of modern transportation. New situations such that infected people may travel not only from the outbreak city to the first destination but also from the first destination to the second one are also crucial for the trend analysis, which has not been considered by the current models. Furthermore, the basic reproduction number $R_0$ is essential in modeling. However, such a parameter can only be calculated based on infected numbers after infection happens. The model can only analyze the effect of our quarantine or intervention policy by heuristically imagining how low the policy can suppress the $R_0$ to without giving any quantifiable index that we can monitoring\cite{Wang2020Cell,wu2020nowcasting}. A systematic fine-grained spatial-temporal model that can better describe the trend of the pandemic and give quantifiable indexes to make policies aim to mitigate pandemic spread is desperately needed. \\
\begin{figure*}[h]
  \centering
  \includegraphics[width=\textwidth]{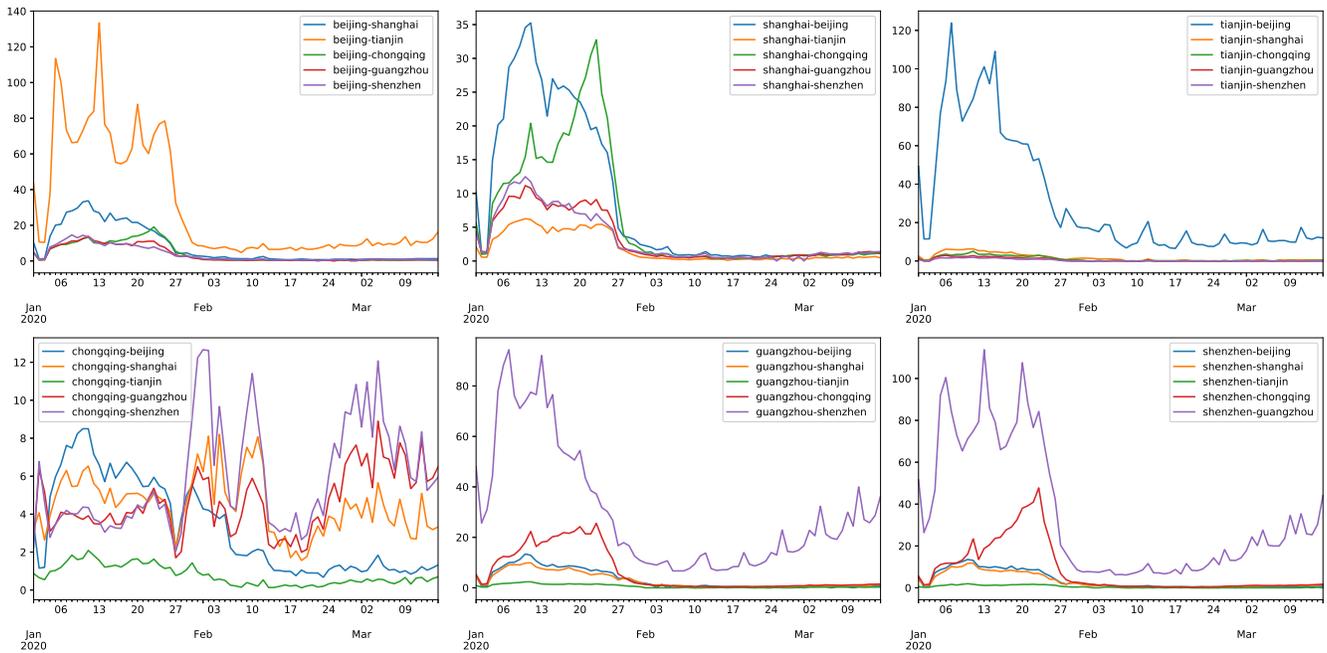}
  \caption{Inter-city mobility index between the major cities in China}
  \Description{inter-city-all index}
  \label{fig:intercity-all}
\end{figure*}

\begin{figure}[h]
  \centering
  \includegraphics[width=\linewidth]{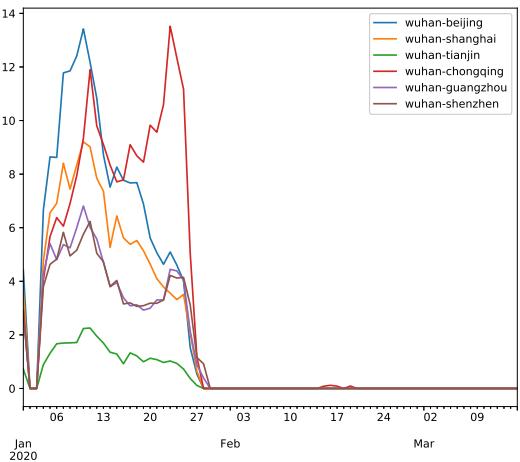}
  \caption{Inter-city mobility index between the outbreak origin Wuhan and other major cities in China}
  \Description{inter-city index}
  \label{fig:intercity}
\end{figure}

To achieve this target, combining mobility data with classical infectious models is an insightful idea. With the prevalence of mobile phones and mobile apps, mobility data provided by large companies gives an accurate and timely overview of the migration of people\cite{Ting2020Tech}. The near real-time and fine-grained data are undoubtedly helpful in building new pandemic modeling frameworks\cite{Aggregated,Ienca2020Digital}. Multiple researches have shown that non-pharmaceutical interventions have obvious effects to mobility activities\cite{arnal2020private,ScienceDataItaly}. For example, Baidu and Tencent provide an apparent migration curve of people between regions\cite{baidu20}, through the data collected from various apps (Figure\ref{fig:intercity}). Moreover, due to the accuracy of positioning and ITS development, the trajectory information generated by people when they travel within the city also allows us to obtain a description of the degree of intra-city activities.(Figure \ref{fig:intracity}) During the pandemic, compared with the fixed and extensive civil aviation railway statistics, mobility data gave a more fine-grained transportation overview from inter-region and intra-region perspectives.\\

\begin{figure}[h]
  \centering
  \includegraphics[width=\linewidth]{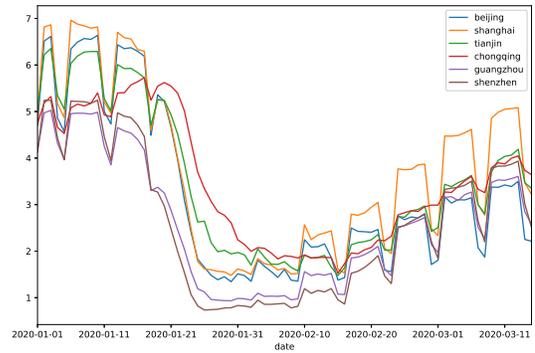}
  \caption{Intra-city mobility index of major cities in China}
  \Description{intra-city index}
  \label{fig:intracity}
\end{figure}

With the considerations list above, we proposed the STSIR framework. A succinct, accurate, quantifiable framework that gives pandemic trend prediction and quantifiable indexes could help mitigate the pandemic spread. Generally, the STSIR model contains a system of differential equations in which we have designed two novel parts to model the spread of infectious diseases intra-regionally and inter-regionally. We set one day as a time segment. We only need the initial information when the outbreak started to give a pandemic evolving curve by fusing the intra-city activity index and the inter-city activity index.

\section{Related works}
With the outbreak of the COVID-19, numerous studies related to such infectious disease are emerged. New models adopted new transmission character of the COVID-19 in order to get better model performance\\
Initially, Wu et al.\cite{wu2020nowcasting} integrated inter-regional traffic originated from Wuhan with SEIR model to catch the pandemic's transmission character. However, inter-regional transmission between other big cities has been ignored. Such a model has not considered latent infection transmission between big cities except the initial outbreak point. Followed by this research, Wang et al. \cite{Wang2020Cell} tried to give an analysis of the pandemic by intuitively adjusting the $R_0$. Such a method gives an overview of the pandemic with different $R_0$s. However, no observable indexes related to $R_0$ are mentioned in such a model, which gives difficulties in evaluating our policy's effect. Furthermore, Gatto et al. \cite{Gatto10484} and Zhou \cite{Zhou2020medxivphase} added a new asymptomatic phase into the pandemic evolving procedure. Gatto et al. \cite{Gatto10484} also added pre-symptomatic and hospitalized phases to give a more accurate dynamic of COVID-19. Nevertheless, with the different definitions of the different phases, accurately categorizing patients into phases is quite difficult. Lau et al.'s work\cite{lau2020association} tried to use air traffic statistics to model transmission traits. However, the granularity of the statistic data limits the model performance. Ting et al. \cite{Ting2020Tech}, Ienca et al.\cite{Ienca2020Digital} and Buckee et al.\cite{Aggregated} pointed out the importance of the mobility data in fighting COVID-19. With such insight, combined with mobility data, Jia et al. \cite{Jia2020NatureST} gave out a model that could help analyze the risk distribution of the pandemics. Gatto et al. \cite{Gatto10484} and Wu et al. \cite{wu2020nowcasting} also enhance their works with inter-regional mobility data. Nevertheless, the mobility activities within each region are not considered thoroughly.\\

\section{Proposed Model}
\subsection{Preliminaries}
In this subsection, we present the fundamental concepts about basic epidemiology models. We start with pure SI and SIR models. Moreover, we analyzed its usefulness in today's COVID-19 cases. As the analysis goes on, we derive our new STSIR framework from the basic models.\\
\textbf{Susceptible-Infected (SI) model.}\quad The most basic epidemiology model is Susceptible-Infected(SI) model. As shown in figure \ref{fig-SI}, there are two transiting phases in the dynamic model.

\begin{figure}[h]
  \centering
  \includegraphics[width=0.5\linewidth]{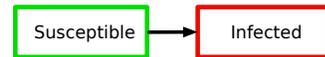}
  \caption{The transiting phases of SI model}
  \Description{SI framework}
  \label{fig-SI}
\end{figure}
And the whole dynamics will be running based on the following equations:

\begin{numcases}{}
	\frac{\partial S}{\partial t}=-\frac{R_{0}IS}{D_{E}N}\\
	\frac{\partial I}{\partial t}=\frac{R_{0}IS}{D_{E}N}
\end{numcases}
with initial condition:
\begin{numcases}{}
	S(0)=N-I_0\\
	I(0)=I_0
\end{numcases}

\textbf{Susceptible-Infected-Removed (SIR) model.}\quad A transformation of the SI model is Susceptible-Infected-Removed (SIR) model. As shown in the figure \ref{fig-SIR} there are three transiting phases in the dynamic model. A new \textbf{Removed} phase is included to describe the people who are recovered, dead, or quarantined, which do not affect the infecting procedure. 

\begin{figure}[h]
  \centering
  \includegraphics[width=0.8\linewidth]{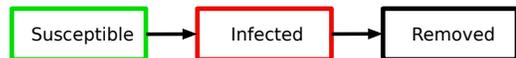}
  \caption{The transiting phases of SIR model}
  \Description{SIR framework}
  \label{fig-SIR}
\end{figure}
And the whole dynamics will be running based on the following equations:
\begin{numcases}{}
	\frac{\partial S}{\partial t}=-\frac{R_{0}IS}{D_{E}N}\\
	\frac{\partial I}{\partial t}=\frac{R_{0}IS}{D_{E}N}-\gamma_{1}I\\
	\frac{\partial R}{\partial t}=\gamma_{1}I
\end{numcases}
with initial condition:
\begin{numcases}{}
	S(0)=N-I_0\\
	I(0)=I_0\\
	R(0)=0
\end{numcases} 
\textbf{Other model with new phases}\quad During this COVID-19 pandemic, clinical reports \cite{Li2020NEJMEarly} have shown that people may show no symptoms or just mild symptoms after they get infected. However, such people also get the ability to infect others. This kind of character makes this virus can be easily spread widely. As shown in figure \ref{fig-SIRBig}, numerous scholars added new phases to enrich the dynamic model. \\
\begin{figure}[h]
  \centering
  \includegraphics[width=\linewidth]{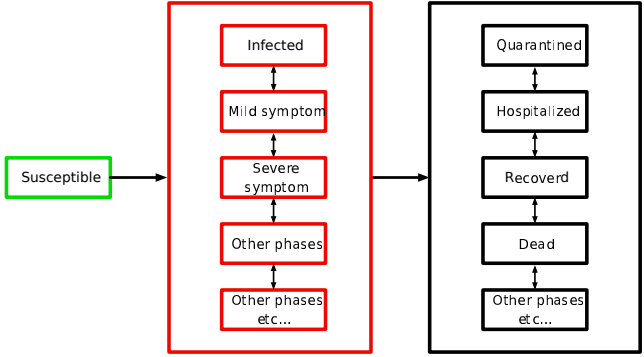}
  \caption{The transiting phases of SIR+ model}
  \Description{SIR+ framework}
  \label{fig-SIRBig}
\end{figure}
Although these models increased the granularity of the transiting phases due to the implicit and different definitions of each phase, clearly categorized patients to the model phases are tricky. To get a clear definition of each phase and a succinct model, we choose to adopt the SIR model as our primary model and focus more on the pandemic's transmission character. 
\subsection{Spatial-Temporal-Susceptible-Infected-Removed (STSIR) Model}

\begin{figure}[h]
  \centering
  \includegraphics[width=\linewidth]{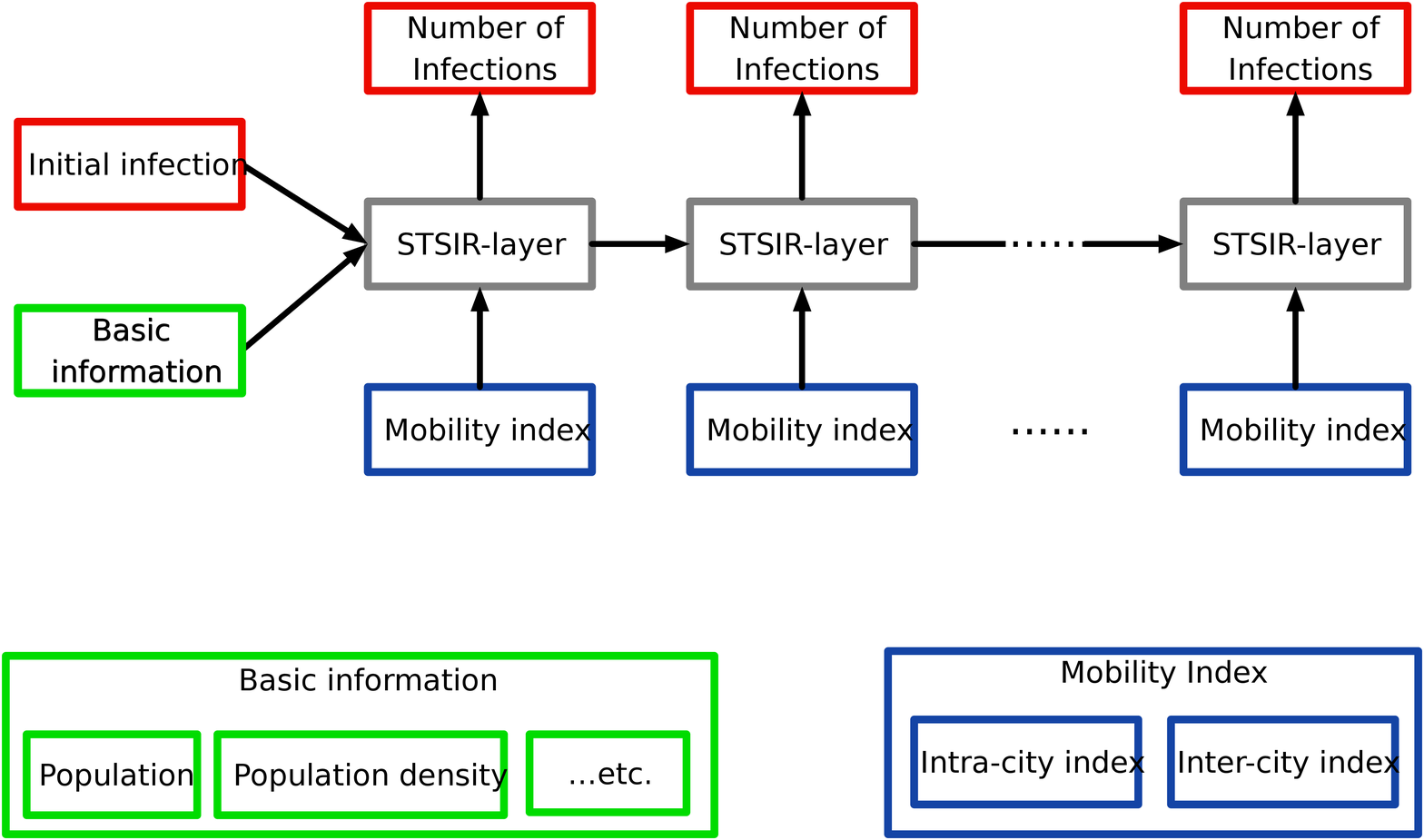}
  \caption{Whole framework of the STSIR system}
  \Description{STSIR framework}
  \label{framework}
\end{figure}
\textbf{STSIR Framework}\quad
As shown in figure \ref{framework}, the whole framework is consist of several \textbf{STSIR Layer}s. Each layer models the daily infection and transmission procedure of the pandemic. With an input of the initial infection and necessary information, the framework can give out the simulated infection numbers based on the corresponding mobility index.\\
\textbf{Mobility Index}\quad 
To better describe the pandemic's transmission character, We adopt two kinds of mobility index in the model. \textbf{Intra-city} index and \textbf{Inter-city} index. The intra-city index describes the activity density within the city. The inter-city index describes the movement between regions. As shown in figure \ref{fig:intercity}, since the government shut down Wuhan, the origin of the outbreak, the inter-city mobility originated from Wuhan to other big cities in China dropped from a high value to near 0. However, there is still transmission between other big cities, as we can see from figure \ref{fig:intercity-all}, which means that there are possibilities that those who get infected can move to multiple big cities and seed the transmission of the pandemic there.\\
Furthermore, as we can see from figure \ref{fig:intracity}, the intra-city index has dropped dramatically when quarantine and stay-home policy begins. However, there are still low-level activities within the city, which give chances of the infection. So both inter-city and intra-city index should be included in our model.\\
\textbf{STSIR Layer}\quad
After considering the problems list above we propose the \textbf{STSIR Layer} which contains a dynamic model shown in follow equations:
\begin{numcases}{}
\begin{split}
\frac{\partial \boldsymbol{S(t)}}{\partial t}=&-\delta_1\cdot\mathbb T_I(t)\frac{\boldsymbol{I(t)\odot S(t)}}{\boldsymbol{N}}\\
&+\delta_2\cdot\boldsymbol{\mathbb T_O (t)}\frac{\boldsymbol{S(t)}}{\boldsymbol{N}}-\delta_2\cdot( \boldsymbol{\mathbb T_O (t)}^T\vec{1})\odot \frac{\boldsymbol{S(t)}}{\boldsymbol{N}}
\end{split}
\label{STSIR system s}
\\
	\begin{split}
	\frac{\partial \boldsymbol{I(t)}}{\partial t}=& \delta_1\cdot\mathbb T_I(t)\frac{\boldsymbol{I(t)\odot S(t)}}{\boldsymbol{N}}-
	\gamma_{1}\boldsymbol {I(t)}\\
	&+\delta_2\cdot\boldsymbol{\mathbb T_O (t)}\frac{\boldsymbol{I(t)}}{\boldsymbol{N}}-\delta_2\cdot( \boldsymbol{\mathbb T_O (t)}^T\vec{1})\odot \frac{\boldsymbol{I(t)}}{\boldsymbol{N}}
	\end{split}
\\
	\frac{\partial \boldsymbol{R(t)}}{\partial t}=\gamma_{1} \boldsymbol{I(t)}
\end{numcases}
with initial condition:
\begin{numcases}{}
	\boldsymbol{S(0)}=\boldsymbol{N}-\boldsymbol{I_0}\\
	\boldsymbol{I(0)}=\boldsymbol{I_0}\\
	\boldsymbol{R(0)}=\boldsymbol{0}
	\label{STSIR system e}
\end{numcases}
Given $k$ cities in the system, here $\boldsymbol{S(t)}, \boldsymbol{I(t)}, \boldsymbol{R(t)}$ are the system susceptible, infectious and recovered vectors with the shape of $[k,1]$ at time $t$ respectively. $\boldsymbol{N}$ denotes the population vector of each city with the shape of $[k,1]$. $\delta_1$ represents the city-wise multiplier for the contribution of intra-city activity to the pandemic transmission. $\delta_2$ is a unified multiplier to denote the transmission between multiple cities. $\gamma_1$ denotes the removing coefficient, which shows the removing rate of the infectious people. $\mathbb T _I(t)$ and $\mathbb T _O(t)$ are Intra-city mobility index with shape the of $[k,1]$ at time $t$ and Inter-city mobility index with the shape of $[k,k]$ at time $t$ respectively. Detailed meanings can be checked from table \ref{tab1}. \\

\begin{table*}
\caption{Notations}
\begin{tabular}{cc}
\toprule
Notations & Meanings\\
\midrule
$N$& Total number of people\\
$\mathbb T _I(t)$& Intra-city mobility index at time $t$ \\
$\mathbb T _O(t)$& Inter-city mobility index at time $t$ \\
$S(t)$& Susceptible people at time $t$ \\
$I(t)$& Infectious people at time $t$ \\
$R(t)$& Recovered people at time $t$ \\
$\gamma_1$& Removing coefficient denoting the removing rate of infected people\\
$\delta_1$ &  Multiplier for intra-city mobility index \\
$\delta_2$& Multiplier for inter-city mobility index \\
\bottomrule
\end{tabular}
\label{tab1}
\end{table*}
Now we show the derivation of such system models. For each city $i$, we consider the transition procedure for each phase, respectively. For the susceptible phase, the change of the number can be consist of three parts:
\begin{itemize}
  \item The people who were susceptible get infected within the city
  \item The people who are susceptible coming from other cities
  \item The people who are susceptible moving to other cities
\end{itemize}
 With the same idea, we can know the four parts of changes of the infectious number: 
 \begin{itemize}
  \item The people who were susceptible get infected within the city
  \item The people who are infectious coming from other cities
  \item The people who are infectious going to other cities
  \item The people who are quarantined, dead, or recovered thus removed from the system.
\end{itemize}
Then for each city $i$, we can have a single dynamic system as follows:
\begin{numcases}{}
\begin{split}
\frac{\partial S(t)_i}{\partial t}=&-\delta_1\cdot\mathbb T_I(t)_i\frac{I(t)_iS(t)_i}{N_i}\\
&+\delta_2S(t)_i\sum_{j=1}^k\mathbb T_O (t)_{ij}-\delta_2S(t)_i\sum_{j=1}^k\mathbb T_O (t)_{ji}
\end{split}
\label{single s}
\\
	\begin{split}
	\frac{\partial I(t)_i}{\partial t}=&\delta_1\cdot\mathbb T_I(t)_i\frac{I(t)_iS(t)_i}{N_i}-\gamma_{1}I(t)_i\\
	&+\delta_2I(t)_i\sum_{j=1}^k\mathbb T_O (t)_{ij}-\delta_2I(t)_i\sum_{j=1}^k\mathbb T_O (t)_{ji}
	\end{split}
\label{single i}
\\
	\frac{\partial{R(t)_i}}{\partial t}=\gamma_{1}I(t)_i
\label{single r}
\end{numcases}
with initial condition
\begin{numcases}{}
	S(0)_i=N_i-I_{0,i}\\
	I(0)_i=I_{0,i}\\
	R(0)_i=0
\end{numcases} 
And we can aggregate $k$ cities with above dynamics together to get the format of the dynamics shown in equation (\ref{STSIR system s}) to (\ref{STSIR system e})
\begin{theorem}
Each STSIR Layer is a \textbf{Self-consistent} system, which means there are no new members introduced to the system. The total population of the whole system always remains the same.
\end{theorem}
\begin{proof}
	For each layer at time $t$ we can add (\ref{single s}), (\ref{single i}) and (\ref{single r}) together. We can get:
	\begin{equation}
	\begin{split}
		\triangle N(t)_i=&\delta_2S(t)_i\sum_{j=1}^k\mathbb T_O (t)_{ij}-\delta_2S(t)_i\sum_{j=1}^k\mathbb T_O (t)_{ji}\\
			&+\delta_2I(t)_i\sum_{j=1}^k\mathbb T_O (t)_{ij}-\delta_2I(t)_i\sum_{j=1}^k\mathbb T_O (t)_{ji}
	\end{split}
	\end{equation}
And for each $\triangle N(t)_i$ we can add them together for k cities. We have:
 	\begin{equation}
	\begin{split}
	\sum_{i=1}^kN(t)_i=&\delta_2S(t)_i\sum_{i=1}^k(\sum_{j=1}^k\mathbb T_O (t)_{ij}-\sum_{j=1}^k\mathbb T_O (t)_{ji})\\
			&+\delta_2I(t)_i\sum_{i=1}^k(\sum_{j=1}^k\mathbb T_O (t)_{ij}-\sum_{j=1}^k\mathbb T_O (t)_{ji})\\
			=&\delta_2S(t)_i(\sum_{i=1}^k\sum_{j=1}^k\mathbb T_O (t)_{ij}-\sum_{i=1}^k\sum_{j=1}^k\mathbb T_O (t)_{ji})\\
			&+\delta_2I(t)_i(\sum_{i=1}^k\sum_{j=1}^k\mathbb T_O (t)_{ij}-\sum_{i=1}^k\sum_{j=1}^k\mathbb T_O (t)_{ji})\\
			=&\delta_2S(t)_i*0+\delta_2I(t)_i*0\\
			=&0
	\end{split}
	\label{proof e}
	\end{equation}
	Based on equation (\ref{proof e}), the inter-city transmission mechanism implement no change of the total population.
\end{proof}	
\textbf{Parameter Learning}\quad
There are several parameters in the STSIR model $\theta= \{\delta_1,\delta_2,\gamma_1 \}$. Because of the asymptomatic state of the COVID-19 and the complexity of the case test, we consider the official announced cases is a delayed expression of the genuine situation. Instead of fitting the officially confirmed curve, we aim to train the model to be able to give a trend of the pandemic by simulating. Given two time sequences, $I_s(t)$ and $I_t(t)$, which denote simulated numbers of daily infections and officially confirmed numbers of daily infections respectively, we want to minimize the gap between total number of cases by simulation and the real confirmed numbers after the test delay, which can be defined as following optimization problem:
\begin{equation}
\underset{\theta}{minimize}\quad Loss(\int_{0}^{t-delay}I_{s}(\theta,t)dt-\int_{0}^{t}I_t(t)dt)
\label{mini}
\end{equation}
To get the parameters of such a multi-system model, we developed a \textbf{Multi-Stage Simulated Annealing} algorithm.  Simulated Annealing \cite{kirkpatrick1983optimization} algorithm has been used for multiple optimization tasks\cite{suman2006survey}, including both single objective and multiple objective optimization problems. It has shown its effectiveness in optimizing parameters. However, traditional simulated annealing is not suitable to optimize a dynamic model with multiple inter-reacting systems. So based on our problem, we use \textbf{Multi-Stage Simulated Annealing} algorithm to find our optimal parameters for the model. The algorithm is described in Algorithm 1.
\begin{algorithm}[]
\label{algorithm}
\DontPrintSemicolon
  
  \Parameters{$\theta= \{\boldsymbol{\delta_1},\delta_2,\boldsymbol{\gamma_1} \}$}
  \KwData{Training set $x$}
  Initializing Parameters..\\
    \While{(current epoch\textless epochs) or (other ending conditions)}
	{
	$loss1=loss(model(x))$\\
	\tcp{Stage 1 update intra-city parameters}
	\For{each $p\in \{\boldsymbol{\delta_1},\boldsymbol{\gamma_1}\}$} 
	{
  	$p_{delta}=uniform(-0.5,0.5)$\\
  	$p=p+p_{delta}*lr_{stage}$ \tcp {randomly update parameters}
 	
  	$loss2=loss(model_{new}(x))$ \tcp {Accept the new parameters if $loss2<loss1$, accept the new parameters with probability if $loss2>loss1$}
	\If {$(loss2>loss1)$}
	{
	\If {$uniform(0,1)>e^{\frac{loss2-loss1}{temp}}$}
		{
  		restore $model.oldparams$\tcp {restore old parameters if update is bad}
		}
	}
	}
	
	\tcp{Stage 2 update inter-city parameter}
	$p=\delta_2$\\
  	$p_{delta}=uniform(-0.5,0.5)$\\
  	$p=p+p_{delta}*lr_{stage}$ \tcp {randomly update parameters}
 	
  	$loss2=loss(model_{new}(x))$ \tcp {Accept the new parameters if $loss2<loss1$, accept the new parameters with probability if $loss2>loss1$}
	\If {$(loss2>loss1)$}
	{
	\If {$uniform(0,1)>e^{\frac{loss2-loss1}{temp}}$}
		{
  		restore $model.oldparams$\tcp {restore old parameters if update is bad}
		}
	}
	
  $temp=temp*coolingrate$\tcp{cooling temperature}
	}

\caption{Multi-Stage Simulated Annealing}
\end{algorithm}

\section{Experiments}
To evaluate the effectiveness of STSIR model, we conduct experiments on with Baidu Mobility dataset \cite{baidu20} and China pandemic dataset provided by Dingxiangyuan \cite{dxy}. We will then provide a brief data description in Section \ref{dataDes}. In Section \ref{setting}, we will analyze the data we have, and illustrate our experiment settings. In Section \ref{predict}, we demonstrate the strong predicting power of our model. In section \ref{policy}, we analyze how great the model can help us analyze the effect of different policies in pandemic transmission.

\subsection{Data Description}
\label{dataDes}

\textbf{Baidu Mobility Dataset}\quad
Baidu Mobility dataset \cite{baidu20} is provided by the Baidu company. Baidu search engine, as well as its various applications, owns over 1 billion users in China. Every time an application is launched, Baidu collects the IP address of the end device. Thus, Baidu could know how the crowd moves from one place to another based on the changes of the IP address. The information is represented as a mobility index, of which a higher number means a more frequent movement of the crowd. The data is fully anonymized and collected according to the terms and conditions of Baidu for privacy concerns.\\
The dataset contains two kinds of mobility data of 388 administrative areas. The administrative areas are four municipalities, 343 prefecture-level cities, and 41 province directly governed counties (including Hong Kong, Macau, and Taiwan). The two kinds of mobility data, intra-city index and inter-city index, are two float numbers representing the activity strength within the city and between cities per day, respectively. The record ranges from 2020-01-01 to 2020-03-15. Since the range covers the Chinese New Year, the dataset also provides a historical intra-city index of the same period based on Lunar Calendar for reference. Table 1 summarizes the statistics of the dataset.

\newcommand{\tabincell}[2]{\begin{tabular}{@{}#1@{}}#2\end{tabular}}
\begin{table}
  \caption{Dataset Description of Baidu Mobility Dataset}
  \label{tab:baidu_mobil}
  \begin{tabular}{ccl}
    \toprule
    Name&Value\\
    \midrule
    \# Administrative areas & 388\\
    Internal flow time range & 2020-01-01 ~ 2020-03-15\\
    \# Internal flow Records & 75\\
    Migration flow time range & \tabincell{c}{2019-01-12 ~ 2019-04-12 \\  2020-01-01 ~ 2020-03-15}\\
    \# Migration flow Records & 128816 \\
  \bottomrule
\end{tabular}
\end{table}

\begin{figure}[]
  \centering
  \includegraphics[width=\linewidth]{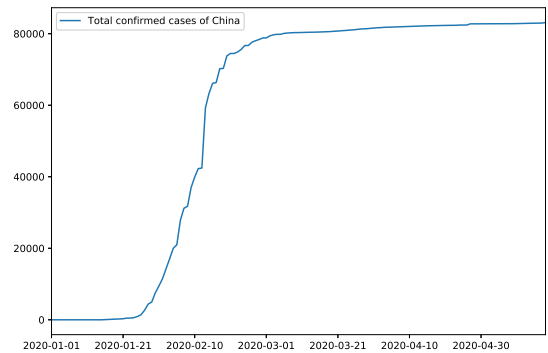}
  \caption{The official confirmed cases of COVID-19 in China.}
  \Description{All china}
  \label{China all}
\end{figure}

\textbf{DingXiangYuan COVID-19 Dataset}\quad
Dingxiangyuan\cite{dxy} COVID-19 dataset contains COVID-19/2019-nCov time series infection data in China. The data source is Ding Xiang Yuan. Ding Xiang Yuan is a professional medical website in China, and its data is authorized and used by mainstream media in China, such as CCTV and China Daily.
The dataset contains the overall infected number in 388 administrative areas, including four municipalities, 343 prefecture-level cities, and 42 province directly governed counties. The time ranges from 2020-01-01 to 2020-05-19. The data, which is an integer number, presents the overall infected number per day per administrative area. Table 2 summarizes the statistics of the dataset.

\begin{table}
  \caption{Dataset Description of DXY COVID-19}
  \label{tab:baidu_DXY}
  \begin{tabular}{ccl}
    \toprule
    Name&Value\\
    \midrule
    \# Administrative areas & 388\\
    Internal flow time range & 2020-01-01 ~ 2020-05-19\\
    \# Internal flow Records & 139\\
    Average Increase per day nationwide & 597.95\\
  \bottomrule
  
\end{tabular}
\end{table}

\subsection{Experiments settings}\label{setting}
As China has conducted a strict policy of distancing and quarantine, the whole pandemic in China become relatively stable (Figure \ref{China all}) during late February, which enables us to test our model by training our model at an early stage, and then feed it with the mobility index to give the simulated final scale of the pandemic. We will calculate the MAE (Mean Absolute Error) with the genuine final scale of the pandemic by:
\begin{equation}
	MAE = (\frac{1}{n})\sum_{i=1}^{k}\left | y_{i} - x_{i} \right |
\end{equation}

\begin{figure*}[htbp]
  \centering
  \includegraphics[width=\textwidth]{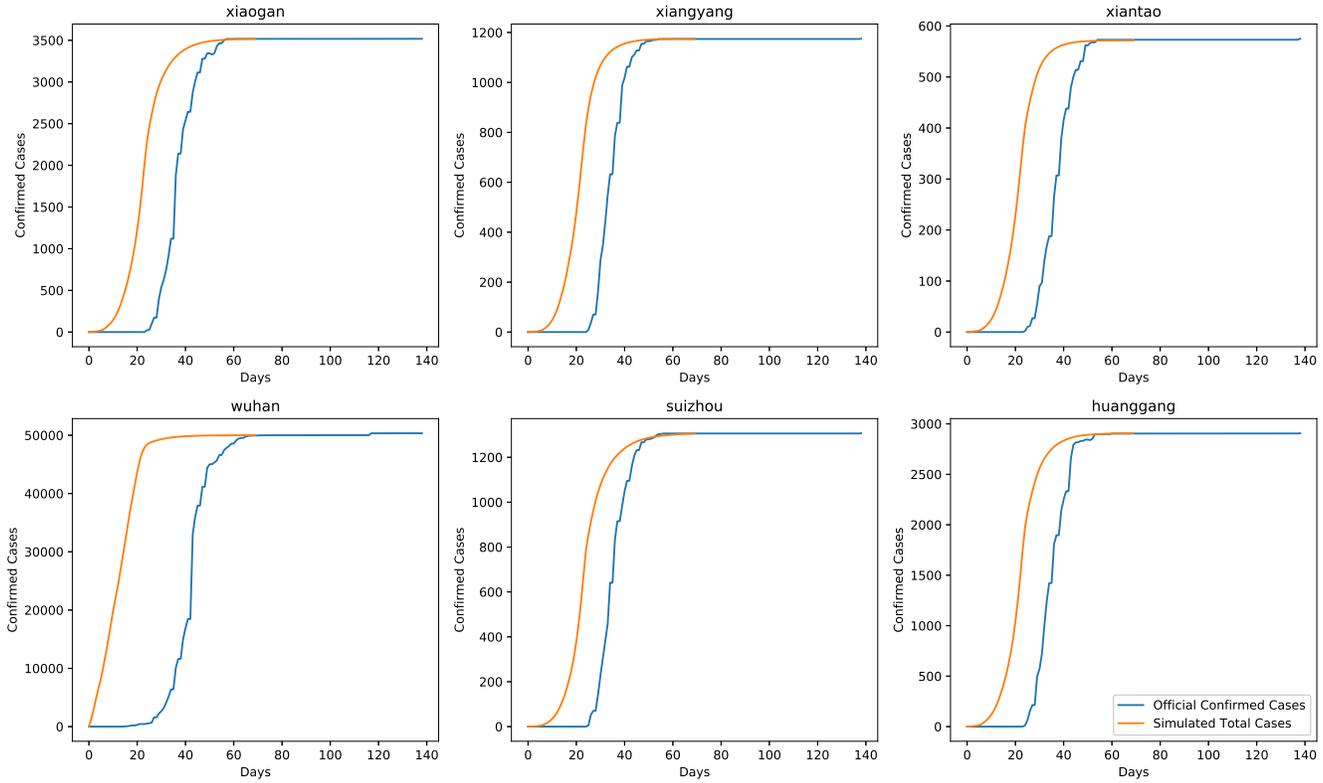}
  \caption{Prediciton results of the STSIR model in major cities in Hubei Province.}
  \Description{Prediction index}
  \label{fig:finalall}
\end{figure*}

\begin{table}
  \caption{Prediction Accuracy}
  \label{tab:accuracy}
  \begin{tabular}{ccl}
    \toprule
    Framework &Overall MAE\\
    \midrule
    STSIR & 7.76\\
  \bottomrule
  
\end{tabular}
\end{table}

\subsection{Prediction Performance}\label{predict}
In this section, we demonstrate how well our model gives the final predictions of the real data. Moreover, the experiments show the model performs well in giving predictions of the total confirmed cases. As shown in the optimization problem (\ref{mini}), we want the model can give the current transmission situation based on the concurrent confirmed cases are given. Because of the asymptomatic character of the COVID-19, there always a delay between the tested confirmed cases and the real infected cases. During the experiment, we set the outbreak that has been passed for 50 days, and the delay between the test and real infection is 20 days. We then train the model at this stage and try to give the total confirmed case at the 70 days, based on real intra-city activity and inter-city activity. As shown in \ref{tab:accuracy}, the STSIR model achieves an overall MAE of 7.76 in predicting the final scale of the pandemic, which denotes its prediction power. As shown in \ref{fig:finalall}, we can see the gap between the real confirmed cases and the simulated curve. It shows the delay effect of the case test. When quarantine and distancing policy is conducted, the transmission tends to be stopped, and the model can predict the final scale of the pandemic well.
\subsection{Policy Analysis}\label{policy}
In this section, we conduct experiments on what if there is no "quarantine" and "shut down" policies are made how severe the pandemic will go in China. Since the transportation system will be really activated during the Chinese New Year, the pandemic will go really mad if no such policies are made. We insert the mobility index last year in the model to see what will happen in Huber Province. As shown in Figure \ref{fig:finalallhis}, if we do not shut down Wuhan, there will be massive infected populations moving to cities near Wuhan and causing massive infections.
\begin{figure*}[b]
  \centering
  \includegraphics[width=\textwidth]{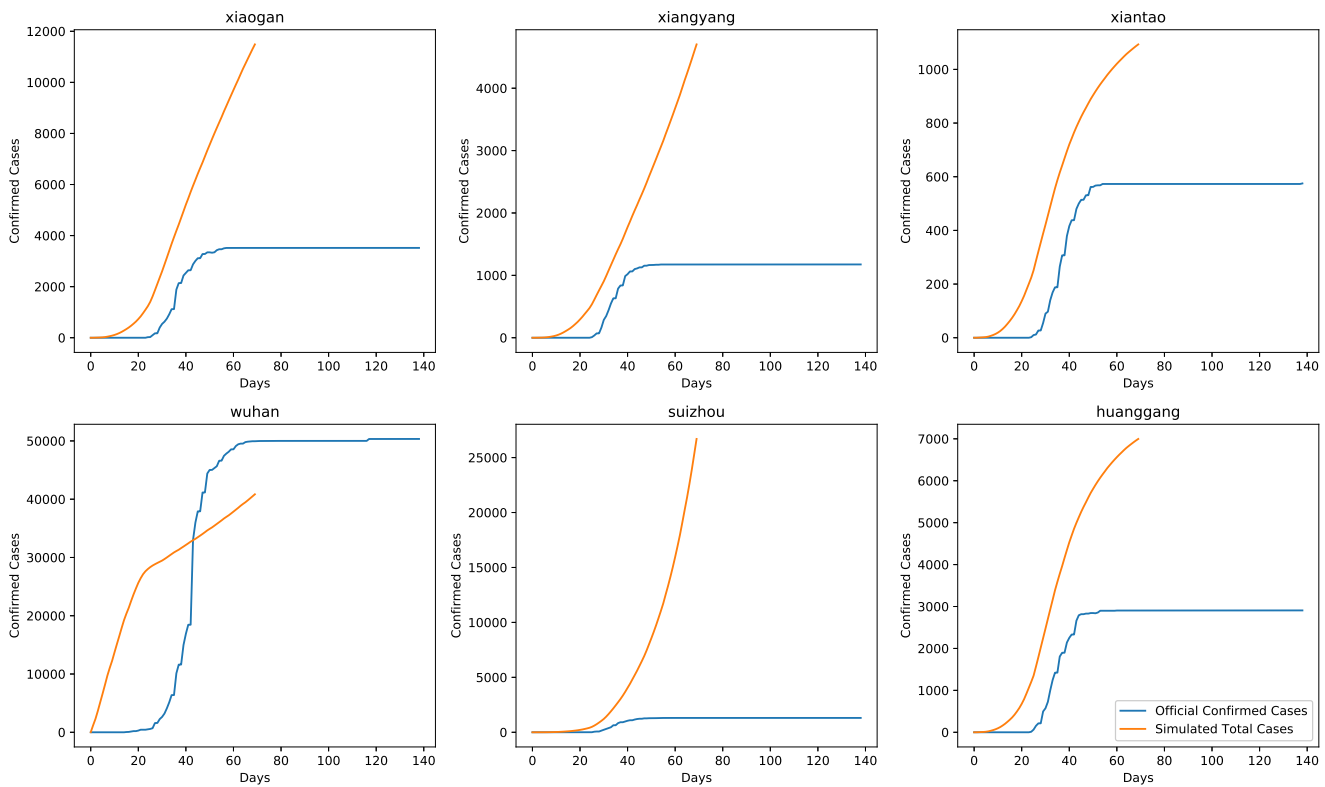}
  \caption{Policy evaluating results of the STSIR model in major cities in Hubei Province.}
  \Description{Evaluate index}
  \label{fig:finalallhis}
\end{figure*}

\section{Conclusion}
To tackle the spread of COVID-19, we introduce the STSIR model to quantify the scale of the pandemic and help to analyze the effect of the policies by introducing different mobility indexes. The system is designed to be self-consistent, and a new MSSA algorithm is introduced to learn the model parameters. Experimental results verify the effectiveness of our study. It shows outstanding predicting power and excellent analysis capability.

\bibliographystyle{ACM-Reference-Format}
\bibliography{bibliography}

\appendix

\end{document}